\documentclass[letterpaper, 10 pt, conference]{ieeeconf}  

\IEEEoverridecommandlockouts                              
\usepackage{amsmath} 
\usepackage{graphicx}
\usepackage{cuted}
\usepackage{flushend}
\usepackage{eucal}
\usepackage{pdflscape}
\usepackage{breqn}
\usepackage{multirow}
\usepackage{slashbox}
\newtheorem{theorem}{Theorem}
\newtheorem{lemma}{Lemma}
\newtheorem{definition}{Definition}

\title{Minimum Energy Analysis for Robust Gaussian Joint Source-Channel Coding with a Square-Law Profile
}

\author{Mohammadamin Baniasadi, and Ertem Tuncel\\
        Department of Electrical and Computer Engineering \\ University of California, Riverside, CA \\ Email: mohammadamin.baniasadi@email.ucr.edu, ertem.tuncel@ucr.edu\\
}

\begin{document}
\maketitle
\thispagestyle{empty}
\pagestyle{empty}

\begin{abstract}

A distortion-noise profile is a function indicating the maximum allowed source distortion value for each noise level in the channel. In this paper, the minimum energy required to achieve a distortion noise profile is studied for Gaussian sources which are transmitted robustly over Gaussian channels. We provide improved lower and upper bounds for the minimum energy behavior of the square-law profile using a family of lower bounds and our proposed coding scheme.

\textit{Index Terms}--Distortion-noise profile, fidelity-quality profile, energy-distortion tradeoff, energy-limited transmission, joint source-channel coding.
\end{abstract}

\section{INTRODUCTION}

Most of emerging wireless applications, such as Internet of things (IoT) and multimedia streaming require lossy transmission of source signals over noisy channels, which is in general a joint source-channel coding (JSSC) problem. Shannon proved the separation theorem which states that in point-to-point scenarios, it is optimal to separate source and channel coding problems. However, in many problems, the optimality of separation breaks down, since JSCC can exploit source correlation to generate correlated channel inputs despite the distributed nature of the encoders, potentially improving the overall performance \cite{c1}-\cite{c5}.

We consider lossy transmission of a Gaussian source over an additive white Gaussian noise (AWGN) channel, where the channel input constraint is not on power and bandwidth, but on energy per source symbol. This approach has drawn much attention recently, see e.g., \cite{c6,c7,c8,c9} as a few  references. Part of the appeal is the simplifications to both achievable schemes and converses as the bandwidth expansion factor approaches infinity [8]. We assume there is no feedback.

It is well-known (for example, see \cite{c6}) that the minimum distortion that can be achieved with energy $E$ when the channel noise variance $N$ is fixed, is given by
\begin{align} \label{D1}
D=\exp(-\frac{E}{N}).
\end{align}
In this paper, a robust setting is considered in which the transmitter does not know $N$, while it is known at the receiver, and it can have any value in the interval (0,$\infty$). The system is to be designed to fulfill with a distortion-noise profile ${\cal D}(N)$ so that it achieves
\begin{align}
D\leq \mathcal{D}(N)\nonumber
\end{align}
for all $0<N<\infty$, while minimizing its energy use. This wide spectrum of noise variances is taken into consideration to account for the scenarios in which absolutely nothing is known about the noise level. For instance, the channel could be suffering occasional interferences of unknown power ($N>0$), although it may be originally of very high quality ($N \approx 0$). There are a wide range of applications in which noise variances are not known. For instance, we can point military situation, indoor fires and emergency conditions. 

In \cite{c10}, it is shown that for the inversely linear profile, uncoded transmission is optimal. Furthermore,it is represented that exponential profiles are not achievable with finite energy. Then, the square-law profile is studied which is somehow combination of linear and exponential profiles and lower and upper bounds have been derived for the minimum achievable energy of the square-law profile. In this paper, we derive improved lower and upper bounds for the minimum energy, and show that the gap between our lower and upper bounds is significantly reduced compared to \cite{c10}. Improving lower and upper bounds and making them as tight as possible helps us to design better systems in practical scenarios by comparing the amount of energy  with these improved theoretical bounds. 

A similar universal coding scenario in the literature is given in \cite{c11}, where a maximum regret approach for compound channels is proposed. The objective in their problem is to minimize the maximum ratio of the capacity to the achieved rate at any noise level. There are other related works including \cite{c12,c13}, and \cite{c14}.

The rest of the paper is organized as follows. The next section is devoted to notation and preliminaries. In Section~III, previous work on lower and upper bounds for the minimum energy is reviewed. In Section~IV, we present our main results, which are improved lower and upper bounds for the square-law profile. Finally, in Section~V we conclude our work and discuss future work.

\section{Notation and Preliminaries}
Suppose that $X^n$ is an i.i.d unit-variance Gaussian source which is transmitted over an AWGN channel $V^m=U^m+W^m$, where $U^m$ is the channel input, $W^m\sim \mathcal{N} (\mathbf{0},N\mathbf{I}_m)$ is the noise, and $V^m$ is the observation at the receiver. We define bandwidth expansion factor $\kappa=\frac{m}{n}$ which can be arbitrarily large, while the energy per source symbol is limited by 
\begin{align}
\frac{1}{n}E(||U^m||^2)\leq E.
\end{align}
The achieved distortion per source symbol is measured as 
\begin{align}
D=\frac{1}{n}E(||X^n-\hat{X}^n||^2)
\end{align}
while $\hat{X}^n$ is the reconstruction at the receiver.

\begin{definition} A pair of distortion-noise profile $\mathcal{D}(N)$ and energy level $E$ is said to be \textit{achievable} if for every $\epsilon>0$, there exists large enough $(m,n)$, an encoder
\begin{align}
f^{m,n}: R^n \to R^m,\nonumber
\end{align}
and decoders
\begin{align}
g_N^{m,n}: R^m 	\to R^n\nonumber
\end{align}
for every $0<N<\infty $, such that
\begin{align}
\frac{1}{n}E\{||f^{m,n}(X^n)||^2\}\leq E+\epsilon\nonumber
\end{align}
and
\begin{align}
\frac{1}{n}E\{||X^n-g_N^{m,n}(f^{m,n}(X^n)+W_N^m)||^2\}\leq \mathcal{D}(N)+\epsilon\nonumber
\end{align}
for all $N$, with $W_N^m$ being the i.i.d. channel noise with variance $N$.
\end{definition}

For given $\mathcal{D}$, the main quantity of interest would be
\begin{align}
E_{min}(\mathcal{D})=\inf \{E:(\mathcal{D},E) \ \textrm{achievable}\}\nonumber
\end{align}
with the understanding that $E_{min}(\mathcal{D})=\infty $ if there is no finite $E$ for which $(\mathcal{D},E)$ is achievable.

In the sequel, it will prove more convenient to use the notation $F=\frac{1}{D}$ and $Q=\frac{1}{N}$, where $F$ and $Q$ standing for signal \textit{fidelity} and channel \textit{quality}, respectively as in \cite{c10}. For any $\mathcal{D}(N)$, we define the corresponding \textit{fidelity-quality profile} as
\begin{align}
\mathcal{F}(Q)=\frac{1}{\mathcal{D}(\frac{1}{Q})}\nonumber
\end{align}
and state that $(\mathcal{F},E)$ is achievable if and only if $(\mathcal{D},E)$ is achievable according to \textit{Definition 1}. $E_{min}(\mathcal{F})$ is similarly defined.
\vspace{3mm}
\section{Previous Work}
\subsection{A Family of Lower Bounds on $E_{min}(\mathcal{D})$}
In \cite{c10}, the authors used the connection between the problem and lossy transmission of Gaussian sources over Gaussian broadcast channels where the power per channel symbol is limited and the bandwidth expansion factor $\kappa$ is fixed. More specifically, they employed the converse result by Tian \textit{et al.} \cite{c15}, which is a generalization of the 2-receiver outer bound shown by Reznic \textit{et al.} \cite{c16} to $K$ receivers, and proved the following lemma.

\begin{lemma}
\label{Lemma1}
For any $K$, $\tau_1 \ge \tau_2 \ge ...\ge \tau_{K-1} \ge \tau_{K}=0$, and $N_1 \ge N_2 \ge ... \ge N_{K}\ge N_{K+1}=0$,
\begin{eqnarray} 
\label{eqtn:lemma1}
E_{min}(\mathcal{D})\!\!\!&\ge & \!\!\!N_1 \log \frac{1+\tau_1}{\mathcal{D}(N_1)+\tau_1}\nonumber\\
\!\!\!&& \!\!\!+\sum_{k=2}^{K} N_k \log \frac{(1+\tau_k)(\mathcal{D}(N_k)+\tau_{k-1})}{(1+\tau_{k-1})(\mathcal{D}(N_k)+\tau_{k})}.
\end{eqnarray}
\end{lemma}

\subsection{Square-Law Fidelity Quality Profiles}
In \cite{c10}, the authors focused on $\mathcal{F}(Q)=1+\alpha Q^2$ for some $\alpha >0$ and analyzed the lower and upper bounds for $E_{min}{(\mathcal{F})}$.
\subsubsection{Lower Bound for $E_{min}(\mathcal{F})$}
Invoking \textit{Lemma 1} by properly choosing $\tau_k$ and $N_k$ in (\ref{eqtn:lemma1}), the following theorem was obtained in \cite{c10}.

\begin{theorem}
For a fidelity-quality profile $\mathcal{F}(Q)=1+\alpha Q^2$, the minimum required energy is lower-bounded as
\begin{align}
E_{min}(\mathcal{F}) \ge c \ \sqrt[]{\alpha}\nonumber
\end{align}
with
\begin{align}
c=\sum_{k=1}^{\infty}\frac{1}{\sqrt[]{4^k\exp(k)-1}}\approx 0.4507. \nonumber
\end{align}
\end{theorem}

\subsubsection{Upper Bound for $E_{min}(\mathcal{F})$}
Using a scheme first sending the source uncoded, and leveraging the received output as side information for the subsequent digital rounds sending indices of an infinite-layer quantizer, an upper bound for the minimum energy was presented in the following theorem in \cite{c10}.

\begin{theorem}
The minimum required energy for profile $\mathcal{F}(Q)=1+\alpha Q^2$ is upper-bounded as
\begin{align}
E_{min}(\mathcal{F})\leq d \ \sqrt[]{\alpha}\nonumber
\end{align}
with 
\begin{align}
d=2 \ \sqrt[]{\log3-Li_2(-2)}\approx 3.1846 \nonumber
\end{align}
where $Li_2(.)$ is the polylogarithm of order 2 defined as 
\begin{align}
Li_2(z)=-\int_{0}^{1} \frac{\log(1-zu)}{u}du.\nonumber
\end{align}
\end{theorem}
\vspace{5mm}
\section{Our Main Results}
Our main contributions in this paper are tighter lower and upper bounds to the energy for the profile ${\cal F}(Q)=1+\alpha Q^2$.
\subsection{Lower Bound for $E_{min}(\mathcal{F})$ }
We begin with lower bounding $E_{\min}({\cal F})$  by the following theorem.
\begin{theorem}
For a fidelity-quality profile ${\cal F}(Q) = 1 + \alpha Q^2$, the minimum required energy is lower-bounded as 
\[
E_{\min}({\cal F})\geq 0.9057\sqrt{\alpha}. 
\]
\end{theorem}

\begin{proof}

A lower bound on $E_{min}(\mathcal{D})$ follows from  (\ref{D1}). Since for any fixed $N_0$ and $D_0$ the expended energy cannot be lower than $N_0\log\frac{1}{D_0}$, the lower bound is obtained given by
\begin{align}
E_{min}(\mathcal{D})\ge \sup_{N>0} N \log \frac{1}{\mathcal{D}(N)}
\end{align}
or equivalently by
\begin{align}
\label{lower1}
E_{min}(\mathcal{F})&\ge \sup_{Q>0} \frac{\log \mathcal{F}(Q)}{Q}\nonumber\\
&=\sup_{Q>0} \frac{\log(1+\alpha Q^2)}{Q}\nonumber\\
&=\bigg(\sup_{q>0} \frac{\log(1+ q^2)}{q}\bigg)\sqrt{\alpha}
\end{align}
where $q= \sqrt {\alpha} Q$.
By solving (\ref{lower1}) numerically, optimal value of $q$ is $q^*=2.01$ and thus 
\begin{align}
\label{lbound}
E_{min}(\mathcal{F})\ge 0.8047\sqrt\alpha.
\end{align}
Note that (\ref{lower1}) is the special case of \textit{Lemma 1} where $K=1$. Thus, it is reasonable to expect an even better lower bound by increasing $K$. By setting $K=2$ and $\tau_1 = \tau \ge \tau_2=0$, the lower bound is achieved as
{\scriptsize\begin{align}
\label{lower2}
E_{min}(\mathcal{F})&\ge \sup_{Q_2>Q_1>0,\tau>0} \Bigg[\frac{\log (1+\frac{\alpha Q^2_1}{1+\tau(1+\alpha Q^2_1)})}{Q_1}+ \frac{\log (1+\frac{\alpha \tau Q^2_2}{1+\tau})}{Q_2} \Bigg]\nonumber\\
&\ge \Bigg(\sup_{q_2>q_1>0,\tau>0} \bigg[\frac{\log (1+\frac{q^2_1}{1+\tau(1+q^2_1)})}{q_1}+ \frac{\log (1+\frac{ \tau q^2_2}{1+\tau})}{q_2} \bigg]\Bigg)\sqrt{\alpha}
\end{align}}
where $q_1=\sqrt{\alpha}Q_1$ and $q_2=\sqrt{\alpha}Q_2$, respectively.

In order to compute the supremum in (\ref{lower2}), we use the gradient ascent algorithm. As the initial point, we set $q_1=2.01$ and $\tau=0$ which give us the same lower bound (\ref{lbound}) for any arbitrary choice of $q_2$. Starting from this initial point (together with the arbitrary choice $q_2=3$), the algorithm converged to $q^*_1=1.5496$, $q^*_2=5.6679$, $\tau^*=0.1285$, and the corresponding lower bound is achieved as
\begin{align}
\label{tightlowerbound}
    E_{min}(\mathcal{F})\ge 0.9057\sqrt\alpha.
\end{align}
\end{proof}
Comparing \textit{Theorem 1} with \textit{Theorem 3} shows that the lower bound is tightened significantly.
%
\subsection{Upper Bound for $E_{min}(\mathcal{F})$}
To upper bound $E_{min}(\mathcal{F})$, we introduce a $K$-layer coding scheme which has a $K$-layer quantizer and sends the quantization indices using Wyner-Ziv coding, where the $k$th quantization index is to be decoded whenever  $N \leq N_k$ for some predetermined $N_1 \geq N_2 \geq \ldots \geq N_K$.
However, instead of relying on only one uncoded transmission of the source $X^n$ as the generator of the side information at the receiver, we also send quantization errors uncoded after each layer of quantization. In other words, we have $K$ layers of uncoded transmission while in  \cite{c10} the authors only had the uncoded transmission in first layer. 

It is not immediately obvious that this strategy will reduce the total expended energy, because even though the energy needed to convey quantization indices will be reduced because of a richer set of available side information, transmission of the quantization errors themselves consumes additional energy. However, as we show here, the minimum energy needed is indeed reduced compared to the scheme in~\cite{c10}.  
\begin{center}
\begin{table*}[t]
\label{table1}
\centering
\caption { \small Utilization of information in our proposed coding scheme}\label{tab1}
\hspace{0.5 cm}
\begin{tabular}{ |l|c|c|c|c|c| }
\hline
Noise interval  & $N>N_1$ &  $N_1\geq N>N_2$   &  $N_2\geq N>N_3$  &  $...$ & $N_{K}\geq N>N_{K+1}$ \\ \hline
\multirow{7}{*}{Decoded digital information}
 & & & & & \\
 & $-$ & $\hat{S}^n_1$& $\hat{S}^n_1$ &  &$\hat{S}^n_1$\\
 &  &  & $\hat{S}^n_2$& & .\\
 & & & & $...$& .\\
 &  & &  & & .\\ 
 &  & &  & & $\hat{S}^n_{K}$\\ 
 & & & & & \\
 \hline
\multirow{9}{*}{Effective side information}
& & & & & \\
 &$\sqrt{A_0} S^n_0+W^n_{0,N}$ & $\sqrt{A_0} S^n_1+W^n_{0,N}$& $\sqrt{A_0} S^n_2+W^n_{0,N}$ &  &$\sqrt{A_0} S^n_{K}+W^n_{0,N}$\\
 &  &$\sqrt{A_1} S^n_1+W^n_{1,N}$  & $\sqrt{A_1} S^n_2+W^n_{1,N}$ & & .\\
 & & &$\sqrt{A_2} S^n_2+W^n_{2,N}$ & & .\\
 &  & &  & $...$& .\\ 
 &  & &  & & .\\
  &  & &  & & .\\
  &  & &  & & $\sqrt{A_{K}} S^n_{K}+W^n_{K,N}$\\ 
  & & & & & \\
 \hline
\end{tabular}
\end{table*}
\end{center}

The source $X^n$ is successively quantized into source codewords $\hat{S}^n_k$ for $k=1,\ldots,K$, where the underlying single-letter characterization satisfies
\begin{align}
S_k=\hat{S}_{k+1}+S_{k+1}\nonumber
\end{align}
with $S_0=X$ and $\hat{S}_{k+1}\perp S_{k+1}$. 
Each $S_k^n$ for $k=0,1,\ldots,K$ is then sent in an uncoded fashion, i.e., as $\sqrt{A_k}S_k^n$.
For any noise variance $0<N<\infty$ , the received signals will then be given by
\[
Y^n_{i,N}=\sqrt{A_i} S_i^n +W^n_{i,N}
\] 
for $i=0,\ldots,K$.
When $N>N_1$, the $X^n$ will be estimated only by utilizing $Y^n_{0,N}$. On the other hand, when $N_{k+1}<N\leq N_k$, for $k=1,2,\ldots,K$, since the first $k$ layers of quantization indices will already be decoded, the estimation can rely on all 
\begin{align}
\tilde{Y}^n_{i,N}=\sqrt{A_i} S^n_k+W^n_{i,N}\nonumber
\end{align}
as \textit{effective} side information, as all $\hat{S}^n_i$ for $i=1,...,k$ can be subtracted from $X^n$.
The utilization of information in our coding scheme is summarized in TABLE I.

Now, to be able to decode $\hat{S_{k}}$ whenever $N \leq N_k$, it suffices to use a binning rate of 
\begin{eqnarray} 
\label{rk}
R_k &=& I(S_{k-1};\hat{S}_k|\tilde{Y}_{0,N_k},\tilde{Y}_{1,N_k},...,\tilde{Y}_{k-1,N_k})\nonumber\\
&=& I(S_{k-1};\hat{S}_k)-I(\tilde{Y}_{0,N_k},\tilde{Y}_{1,N_k},...,\tilde{Y}_{k-1,N_k};\hat{S}_k)\nonumber\\
&=& h(S_{k-1})-h(S_{k})\nonumber\\
&&-h(\tilde{Y}_{0,N_k},\tilde{Y}_{1,N_k},...,\tilde{Y}_{k-1,N_k})\nonumber\\
&&+h(\tilde{Y}_{0,N_k},\tilde{Y}_{1,N_k},...,\tilde{Y}_{k-1,N_k}|\hat{S}_k)\nonumber\\
&=&\frac{1}{2}\log\frac{\sigma^2_{S_{k-1}}}{\sigma^2_{S_{k}}}-\frac{1}{2}\log\frac{\det\mathbf{\Sigma}_{\mathbf{Y}_k}}{\det\mathbf{\Sigma}_{\mathbf{Y}_k|\hat{S}_k}}
\end{eqnarray}
where 
\[
\mathbf{\Sigma}_{\mathbf{Y}_k} = \mathbf{A}_k\mathbf{\Sigma}_{\mathbf{Z}_k}\mathbf{A}_k^T
\]
with 
 \begin{align}
\mathbf{A}_k=
  \begin{bmatrix}
    \sqrt[]{A_0} & 1 & 0 & 0 & . & . & . & 0  \\
     \sqrt[]{A_1} & 0 & 1 & 0 & . & . & . & 0  \\
    . & . & . & 1 & . & . & . & 0  \\
    . & . & . & . & 1 & . & . & 0  \\
    . & . & . & . & . & 1 & . & 0  \\
    \sqrt[]{A_{k-1}} & 0 & 0 & 0 & . & . & . & 1  \\
    \end{bmatrix}\nonumber
\end{align}
and 
\begin{align}
\mathbf{Z}_k=
  \begin{bmatrix}
    S_{k-1} \\
    W_{0,N_k} \\
    W_{1,N_k}\\
 \vdots\\
     W_{k-1,N_k}
  \end{bmatrix} \; .\nonumber
\end{align}  
Similarly, 
\[
\mathbf{\Sigma}_{\mathbf{Y}_k|\hat{S}_k} = \mathbf{A}_k\mathbf{\Sigma}_{\mathbf{\tilde{Z}}_k}\mathbf{A}_k^T
\]
with 
\begin{align}
  \mathbf{\tilde{Z}}_k=
  \begin{bmatrix}
    S_{k} \\
    W_{0,N_k} \\
    W_{1,N_k}\\
    \vdots \\
         W_{k-1,N_k}
  \end{bmatrix} \; .\nonumber
\end{align}
Since the source and channel noise are independent, both $\mathbf{\Sigma}_{\mathbf{Z}_k}$ and $\mathbf{\Sigma}_{\mathbf{\tilde{Z}}_k}$ are diagonal, and that makes the computation of $\mathbf{\Sigma}_{\mathbf{Y}_k}$ and $\mathbf{\Sigma}_{\mathbf{Y}_k|\hat{S}_k} $ easy.
Specifically, defining the $k\times k$ matrix
\[
\mathbf{G}_k=
  \begin{bmatrix}
    1 & 0 & . & . & . & 0\\
     0 & 0 & 0 & . & .  & 0  \\
    . & 0 & 0 & 0 & . & 0  \\
    . & . & 0 & . & 0 & 0    \\
    . & . & . & . & 0 & 0    \\
    0 & 0 & 0 & 0 & 0 & 0   \\
    \end{bmatrix} ,
\]
one can write 
\[
\mathbf{\Sigma}_{\mathbf{Z}_k}= N_k \mathbf{I}_{k}+(\sigma^2_{S_{k-1}}-N_k)\mathbf{G}_k \; .
\]
and 
\[
\mathbf{\Sigma}_{\mathbf{\tilde{Z}}_k}= N_k \mathbf{I}_{k}+(\sigma^2_{S_k}-N_k)\mathbf{G}_k \; .
\]
We then have
\begin{align}
\mathbf{\Sigma}_{\mathbf{Y}_k}&=\mathbf{A}_k\bigg(N_k \mathbf{I}_{k}+(\sigma^2_{S_{k-1}}-N_k)\mathbf{G}_k\bigg)\mathbf{A}_k^T\nonumber\\&=N_k\mathbf{A}_k\mathbf{I}_{k}\mathbf{A}_k^T+(\sigma^2_{S_{k-1}}-N_k)\mathbf{A}_k\mathbf{G}_k\mathbf{A}_k^T\nonumber\\
&=N_k\mathbf{A}_k\mathbf{A}_k^T+(\sigma^2_{S_{k-1}}-N_k) \mathbf{a}_k  \mathbf{a}_k^T\nonumber \\
&=N_k( \mathbf{a}_k  \mathbf{a}_k^T+\mathbf{I}_k)+(\sigma^2_{S_{k-1}}-N_k)  \mathbf{a}_k  \mathbf{a}_k^T \nonumber \\
 \label{covy}
&=N_k \mathbf{I}_k+\sigma^2_{S_{k-1}} \mathbf{a}_k  \mathbf{a}_k^T \; 
\end{align}
where $\mathbf{a}_k$ is the first column of matrix $\mathbf{A}_k$. Similarly,
\begin{align} \label{covys}
\mathbf{\Sigma}_{\mathbf{Y}_k|\hat{S}_k}=N_k\mathbf{I}_k+\sigma^2_{S_k} \mathbf{a}_k  \mathbf{a}_k^T.
\end{align}
By substituting (\ref{covy}) and (\ref{covys}) in (\ref{rk}), we then get
\begin{align} \label{Rdet}
R_k=\frac{1}{2}\log\frac{\sigma^2_{S_{k-1}}}{\sigma^2_{S_{k}}}-\frac{1}{2}\log\frac{\det(N_{k}\mathbf{I}_{k}+\sigma^2_{S_{k-1}}\mathbf{a}_k \mathbf{a}_k^T)}{\det(N_{k}\mathbf{I}_{k}+\sigma^2_{S_{k}}\mathbf{a}_k  \mathbf{a}_k^T)}.
\end{align}
Using the Matrix Determinant Lemma \cite{c17}, which states for arbitrary invertible $\mathbf{M}$ and column vectors $\mathbf{u}$ and $\mathbf{v}$ that 
\[
\det(\mathbf{M}+\mathbf{u}\mathbf{v}^T) = \det(\mathbf{M}) \cdot (1+\mathbf{v}^T\mathbf{M}^{-1}\mathbf{u})
\]
We can write 
 \begin{align}
R_k&=\frac{1}{2}\log\frac{\sigma^2_{S_{k-1}}}{\sigma^2_{S_{k}}}-\frac{1}{2}\log\frac{(1+\frac{\sigma^2_{S_{k-1}}}{N_k}\mathbf{a}_k^T\mathbf{a}_k)}{(1+\frac{\sigma^2_{S_{k}}}{N_k}\mathbf{a}_k^T\mathbf{a}_k)}\nonumber\\
&=\frac{1}{2}\log\frac{\sigma^2_{S_{k-1}}(1+\frac{\sigma^2_{S_{k}}}{N_k}\mathbf{a}_k^T\mathbf{a}_k)}{\sigma^2_{S_k}(1+\frac{\sigma^2_{S_{k-1}}}{N_k}\mathbf{a}_k^T\mathbf{a}_k)} \nonumber \\
&=\frac{1}{2}\log\frac{\beta_k+Q_kA_{k,\rm{total}}}{\beta_{k-1}+Q_kA_{k,\rm{total}}}  
\end{align}
where $\beta_k=\frac{1}{\sigma^2_{S_{k}}}$, $Q_k=\frac{1}{N_k}$, and 
\[
A_{k,\rm{total}}\stackrel{\Delta}{=}A_0+A_1+...+A_{k-1} \; .
\]

For the digital message, we use the channel with infinite bandwidth and energy $B_k$. Therefore, the rate must not exceed the channel capacity under the noise level $N_k$, i.e., 
\[
\frac{1}{2}\log\frac{\beta_k+Q_kA_{k,\rm{total}}}{\beta_{k-1}+Q_kA_{k,\rm{total}}} \leq C_k=\frac{B_kQ_k}{2}
\]
or equivalently, 
\[
\frac{\beta_k+Q_kA_{k,\rm{total}}}{\beta_{k-1}+Q_kA_{k,\rm{total}}} \leq \exp(B_kQ_k) \; .
\]

When $N_{k+1}<N\leq N_k$, or equivalently when $Q_k\leq Q< Q_{k+1}$, the MMSE estimation boils down to estimating $\tilde {S}_k^n$ using all the available effective side information, that is 
\[
\tilde {S}_k^n=\sum_{i=0}^{k}\lambda_i\tilde{Y}_{i,N}
\]
with appropriate $\lambda_i$ for $i=0,1,\ldots,k$.
Thus, the resultant distortion can be calculated with the help of the Sherman-Morrison-Woodbury identity \cite{c17} as 
\begin{align} \label{D}
D&=\sigma^2_{S_k}-(\sigma^2_{S_k})^2\mathbf{a}_{k+1}^T\left(N \mathbf{I}_{k+1}+\sigma^2_{S_k} \mathbf{a}_{k+1}  \mathbf{a}_{k+1}^T\right)^{-1}\mathbf{a}_{k+1} \nonumber \\
&=\sigma^2_{S_k}-(\sigma^2_{S_k})^2\mathbf{a}_{k+1}^T\left[Q\mathbf{I}_{k+1} - \frac{Q^2 \mathbf{a}_{k+1}\mathbf{a}_{k+1}^T}{\beta_k + Q \mathbf{a}_{k+1}^T\mathbf{a}_{k+1}}\right]\mathbf{a}_{k+1}  \nonumber \\
&=\sigma^2_{S_k}-(\sigma^2_{S_k})^2Q A_{k+1,\rm{total}}\left[1 - \frac{Q A_{k+1,\rm{total}}}{\beta_k + Q A_{k+1,\rm{total}}}\right] \nonumber \\
&=\sigma^2_{S_k}-\sigma^2_{S_k}Q A_{k+1,\rm{total}}\left[\frac{1}{\beta_k + Q A_{k+1,\rm{total}}}\right] \nonumber \\
&=\frac{1}{\beta_k + Q A_{k+1,\rm{total}}}.
\end{align}
Equivalently, the fidelity can be written as 
\begin{align}\label{F}
F(Q)=\beta_k + Q A_{k+1,\rm{total}} \; .
\end{align}
Therefore, $F(Q)$ is an ``inclined'' staircase function with changing slope $A_{k+1,\rm{total}}$ within each $Q_k\leq Q< Q_{k+1}$ as shown in Fig.~\ref{figr:UpperBound}. The beauty of the work is that we deal with linear segments. Thus, our analysis is easily understandable. 
Please note that Fig.~\ref{figr:UpperBound} is different with the figure presented in \cite{c10}. The slope of inclined staircase function is fixed and equal to $E_0$ in \cite{c10}, which is a special case of our work by letting $A_0=E_0$ and $A_i=0$ for $i=1,2,...$ . We are now ready to propose an upper bound on $E_{min}({\cal F})$.
\begin{figure}
  \centering
    \includegraphics[width=0.5\textwidth]{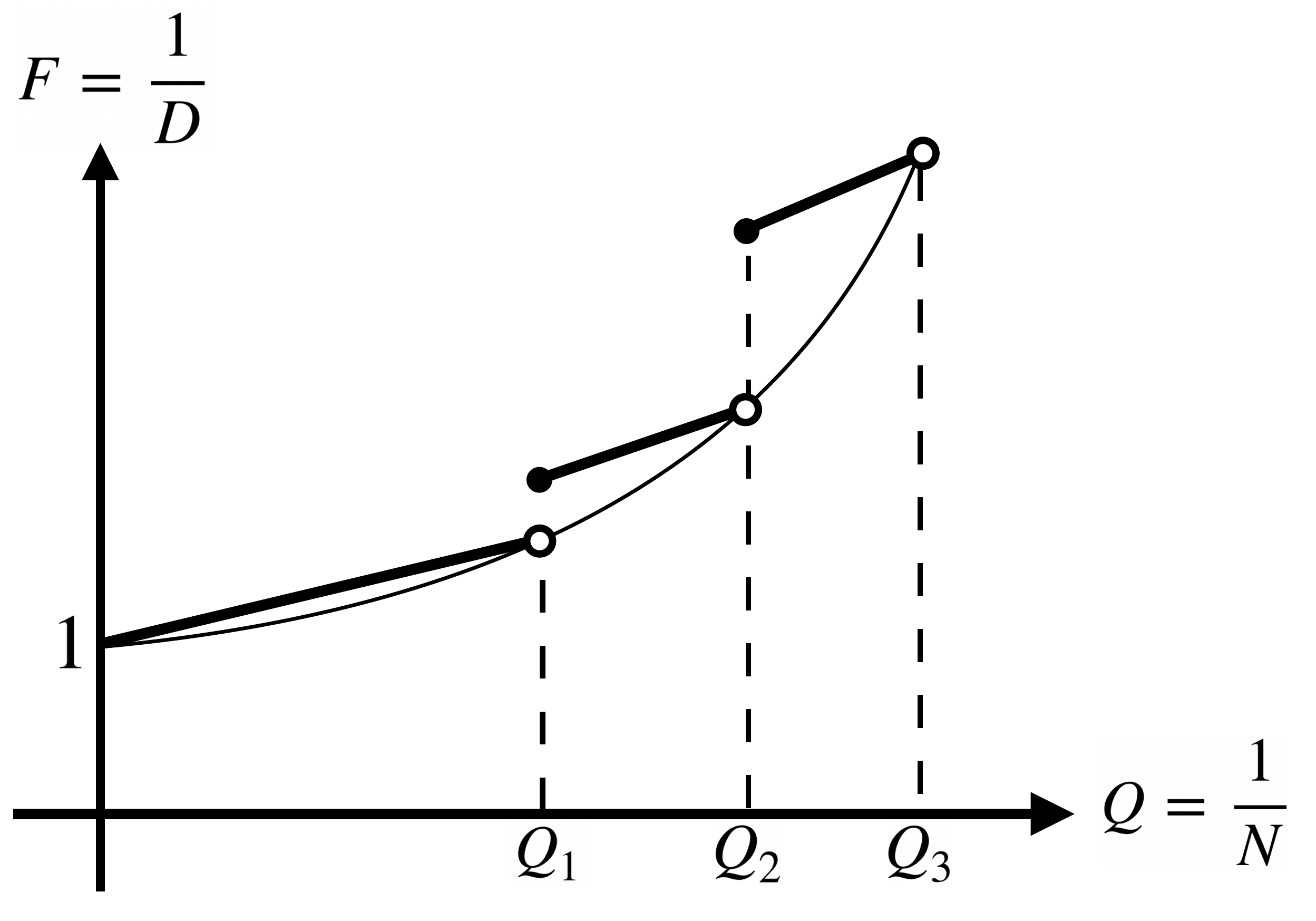}
    \caption{\label{figr:UpperBound}The fidelity-quality tradeoff is always above the profile $F(Q)$, coinciding with it at the jump points $Q_k$.}
\end{figure}

\begin{theorem}
The minimum required energy for profile $F(Q)=1+\alpha Q^2$ is upper-bounded as 
\begin{align}
E_{min}(F)\leq e\ \sqrt[]{\alpha}\nonumber
\end{align}
with $e\approx 2.3203$.
\end{theorem}
\begin{proof}
We will use the scheme described above such that for any $0=Q_0<Q_1<Q_2<...,$ the energy $A_k$ and the source coding parameters $1=\beta_0<\beta_1<\beta_2<...$ will be chosen such that the fidelity-quality tradeoff in (\ref{F}) is always above the profile $F(Q)$, coinciding with it at the jump points $Q_k$, as shown in Fig.~\ref{figr:UpperBound}. In other words,
\begin{align}
A_{k,\rm{total}}Q_k+\beta_{k-1}=1+\alpha Q^2_k
\end{align}
for all $k=1,2,...$ .

Thus, we obtain
\begin{align} \label{betak}
\beta_{k-1}&=1+\alpha Q^2_k-A_{k,\rm{total}} Q_k\nonumber\\
&=1+\alpha Q^2_k-(A_0+A_1+...+A_{k-1})Q_k.
\end{align}
The requirement that $\beta_{k}$ is increasing in $k$ leads to the following constraint
\begin{align} \label{Ak_constraint}
A_0&=\alpha Q_1 \nonumber\\
A_k&<\frac{\alpha (Q^2_{k+1}-Q^2_k)- A_{k,\rm{total}}(Q_{k+1}-Q_k)}{Q_{k+1}},
\end{align}
for all $ k\geq 1$.
\begin{lemma}
\label{lemma2}
For fixed $\alpha$, the choice $Q_k=k\Delta$, $A_0=\alpha\Delta$ and $A_k=d^k\alpha\Delta$ for $k\geq 1$ satisfies (\ref{Ak_constraint}) for any $0<d<1$ and $\Delta >0$.
\end{lemma}
\begin{proof}
Substituting $A_k$ and $Q_k$ in (\ref{Ak_constraint}) yields:
\begin{align} \label{dk}
d^k<\frac{(2k+1)-(\frac{d^k-1}{d-1})}{k+1}.
\end{align}
We use induction to prove (\ref{dk}). For $k=1$, (\ref{dk}) reduces to $d<1$ which is true. Substituting $k=l$, we get to the following:
\begin{align}\label{dn}
d^l(l+1)<(2l+1)-(d^{l-1}+...+d+1)
\end{align}
We assume (\ref{dn}) is true. Now we substitute $k=l+1$ in (\ref{dk}) and have the following:
\begin{align}\label{dn+1}
d^{l+1}(l+2)<(2l+3)-(d^l+d^{l-1}+...+d+1).
\end{align}
In order to complete  the proof, we show (\ref{dn+1}) is true as follows.
First, we multiply both sides of (\ref{dn}) with $d$ and then add $d^{l+1}$ to both sides, yielding
\begin{align}\label{dnn}
d^{l+1}(l+2)<(2l+1)d+d^{l+1}-(d^l+...+d).
\end{align}
Now, it suffices to show the right hand side of (\ref{dnn}) is less than or equal to the right hand side of (\ref{dn+1}), which is the same as
\begin{align}\label{dinequal}
d^{l+1}+(2l+1)d\leq 2l+2.
\end{align}
Since $0<d<1$, we have $(2l+1)d<(2l+1)$ and $d^{l+1}<1$. Thus, (\ref{dinequal}) is valid and the proof of \textit{Lemma 2} is complete.
\end{proof}
By substituting $A_k$ and $Q_k$ in (\ref{betak}), we get
\begin{align}
\beta_{k-1}&=1+\alpha k^2\Delta^2-k\alpha\Delta^2\bigg(\frac{1-d^k}{1-d}\bigg).\nonumber
\end{align}
Letting $K\to\infty$, the total uncoded energy becomes
{\footnotesize\begin{align}
E_{unc}&=\sum_{k=0}^{\infty}\frac{A_k}{\beta_k}\nonumber\\
&=\sum_{k=0}^{\infty}\frac{d^k\alpha\Delta}{1+\alpha (k+1)^2\Delta^2-(k+1)\alpha\Delta^2\bigg(\frac{1-d^{k+1}}{1-d}\bigg)}.
\end{align}}
\begin{center}
\begin{table*}[t]
\label{table2}
\centering
\caption { \small  Comparison between our lower and upper bounds with the bounds in \cite{c10}}
\begin{tabular}{|l||*{5}{c|}}\hline
\backslashbox[48mm]{Minimum Energy values}{$\alpha$}
&\makebox[3em]{1}&\makebox[3em]{10}&\makebox[3em]{100}
&\makebox[3em]{1000}&\makebox[3em]{10000}\\\hline\hline
Lower bound of \cite{c10} & 0.4507 & 1.4252 & 4.5070 & 14.2524 & 45.0700\\\hline
Our lower bound & 0.9057 & 2.8641 & 9.0570 & 28.6407 & 90.5700\\\hline
Upper bound of \cite{c10} & 3.1846 & 10.0706 & 31.8460 & 100.7059 & 318.4600\\\hline
Our upper bound & 2.3203 & 7.3374 & 23.2030 & 73.3743 & 232.0300\\\hline
The lower bound improvement & 0.4550 & 1.4388 & 4.5500 & 14.3884 & 45.5000\\\hline
The upper bound improvement & 0.8643 & 2.7332 & 8.6430 & 27.3316 & 86.4300\\\hline
\end{tabular}
\end{table*}
\end{center}

On the other hand, the total digital energy is 
{\scriptsize \begin{align}
E_{dig}&=\sum_{k=1}^{\infty}B_k\nonumber\\
&=\sum_{k=1}^{\infty}\frac{1}{Q_k}\log\Bigg(1+\frac{\beta_{k}-\beta_{k-1}}{1+\alpha Q^2_{k}}\Bigg)\nonumber\\
&=\sum_{k=1}^{\infty}\frac{1}{k\Delta}\log\Bigg(1+\frac{\alpha\Delta^2\bigg(2k+1-kd^k-(\frac{1-d^{k+1}}{1-d})\bigg)}{1+\alpha k^2\Delta^2}\Bigg).
\end{align}}

Denoting $\Delta=\frac{c}{\sqrt \alpha}$, we then get
{\footnotesize \begin{align}
\label{etot}
&E_{total}=E_{unc}+E_{dig}\nonumber\\
&\leq\sqrt{\alpha}\sum_{k=0}^{\infty}\frac{cd^k}{1+c^2 (k+1)^2-(k+1)c^2\bigg(\frac{1-d^{k+1}}{1-d}\bigg)}\nonumber\\
&+\sqrt{\alpha}\sum_{k=1}^{\infty}\frac{1}{kc}\log\Bigg(1+\frac{c^2\bigg(2k+1-kd^k-(\frac{1-d^{k+1}}{1-d})\bigg)}{1+k^2c^2}\Bigg). 
\end{align}}

In order to minimize the upper bound on total energy, we solve (\ref{etot}) numerically for different values of $c$ and $d$. For optimal values $d^*=0.999$ and $c^*=0.00137$, the upper bound yields
\begin{align}
E_{total}\leq 2.3203\sqrt{\alpha}.\nonumber
\end{align}
\end{proof}

Comparing \textit{Theorem 2} with \textit{Theorem 4}, we notice that the upper bound is improved significantly. 
Note that by setting $d=0$ and $A_0=E_0$ in our work, the method in \cite{c10} is achieved exactly. This is expected, as \cite{c10} is a special case of our work. 

We compare our lower and upper bounds with the bounds in \cite{c10} for some values of $\alpha$ and show our improvements in TABLE 2.
\vspace{2mm}
\section{Conclusions and Future Work}
Minimum energy required to achieve a distortion-noise profile, i.e., a function indicating the maximum allowed distortion value for each noise level, is studied for robust transmission of Gaussian sources over Gaussian channels. 
In order to analyze the minimum energy behavior for the square-law distortion noise profile, the lower and upper bounds were proposed by our coding scheme. We improved both upper and lower bounds significantly.
For future, we are interested to study the distortion-noise profile problem in Multiple Access Channels (MAC). In MAC, instead of having one distortion function, we deal with at least two distortion functions and distortion regions. 


\end{document}